\documentclass[11pt]{article}

\usepackage[T1]{fontenc}
\usepackage[utf8]{inputenc}
\usepackage{amsmath,amssymb,amsthm}
\usepackage{bm}
\usepackage{graphicx}
\usepackage{booktabs,longtable,tabularx,array}
\usepackage{mathrsfs}
\usepackage{newtxtext,newtxmath}
\usepackage[authoryear,round]{natbib}
\usepackage{xcolor}
\usepackage[margin=2.35cm]{geometry}
\usepackage{microtype}
\usepackage[colorlinks=true,linkcolor=blue!50!black,citecolor=blue!50!black,
            urlcolor=blue!50!black]{hyperref}

\newcommand{\khat}{\widehat{\bm k}}
\newcommand{\that}{\widehat{\bm t}}

\newcommand{\A}{\boldsymbol{\mathsf A}}
\newcommand{\I}{\boldsymbol{\mathsf I}}
\newcommand{\tr}{\operatorname{tr}}
\newcommand{\Rea}{\operatorname{Re}}
\newcommand{\Ima}{\operatorname{Im}}
\newcommand{\dd}{\mathrm d}
\newcommand{\e}{\boldsymbol e}

\newcommand{\T}{\mathsf T}
\newtheorem{proposition}{Proposition}[section]
\newtheorem{lemma}{Lemma}[section]
\numberwithin{equation}{section}
\numberwithin{table}{section}
\hypersetup{
  pdftitle={Static compliance and directional instability in indefinite conformation states},
  pdfauthor={Yuan Yu}
}

\title{\bfseries Static compliance and directional instability in indefinite
conformation states}
\author{Yuan Yu\\[0.45em]
\small School of Mathematics and Computational Science, Xiangtan University,
Xiangtan 411105, China\\
\small National Center for Applied Mathematics in Hunan,\\
\small Xiangtan 411105, China\\
\small Hunan Key Laboratory for Computation and Simulation in Science and
Engineering, Xiangtan University,\\
\small Xiangtan 411105, China\\[0.35em]
\small \href{mailto:yuyuan@xtu.edu.cn}{yuyuan@xtu.edu.cn}\\
ORCID: \href{https://orcid.org/0000-0001-5125-2492}{0000-0001-5125-2492}}
\date{}

\begin{document}
\maketitle

\begin{abstract}
A conformation tensor is positive definite for every physically realizable
polymer microstructural state. Numerical discretization can move the
conformation tensor outside the positive-definite domain. This raises a
question: can the least eigenvalue alone identify the first unstable direction?
We answer it by linearizing Oldroyd-B, equilibrium-normalized FENE-P and
Giesekus models about uniform frozen states. All include solvent viscosity and
stress diffusion. We examine every non-zero planar Fourier mode, assuming each
model's uncoupled constitutive tangent is strictly stable. The margin $1+r\chi$
measures the balance between solvent damping and the zero-frequency polymer
response. The complete velocity--conformation system is stable if and only if
this margin is positive. At zero margin, a simple stationary root appears;
finite inertia changes growth rates but not the neutral boundary. At fixed
wavenumber and other parameters, decreasing
$\lambda_1$ identifies the first neutral direction. Oldroyd-B and
equilibrium-normalized FENE-P first become neutral along principal directions;
Giesekus mobility can instead make an oblique direction neutral first. For the
reference case, onset is $\lambda_{1,c}=-1.933$ at $\theta_c=23.94^\circ$,
before the principal-axis prediction. Along this family, the all-direction
threshold approaches $-2.319$ as the other principal stretch grows, whereas the
formal principal-axis extrapolation tends to negative infinity. A
rational-parameter counterexample, matrix spectra and uniform forced-base
calculations test the neutral boundary and both sides. These results concern one
linear, uniform, planar Fourier mode, not nonlinear or inhomogeneous-flow
stability. Within this scope, onset depends not on indefiniteness alone but also
on constitutive-tangent geometry and wavevector direction.
\end{abstract}


\section{Introduction}
\label{sec:intro}

The conformation tensor $\A$ is a second moment of polymer configurations and is therefore symmetric positive definite (SPD) for every realizable microstructural state. Its eigenvalues measure mean-square stretch, so a negative eigenvalue has no corresponding polymer ensemble. Differential Maxwell-type models preserve this constraint in exact evolution \citep{DupretMarchal1986,Hulsen1990}, and matrix-logarithm formulations exploit it to keep computations inside the SPD cone \citep{FattalKupferman2004}. A discretization, however, can produce an indefinite numerical conformation state near poorly resolved stress gradients. Recent simulations illustrate that such formally inadmissible states may coexist with apparently converged observables \citep{CapocciEtAl2026}; that observation motivates, but does not validate, a local analysis of what the continuum symbol does when an outside-cone state is supplied to it. Physical realizability, numerical persistence and spectral response are distinct. We address only the last, while retaining the first as an explicit boundary on interpretation.

For a solvent-free Maxwell liquid, an indefinite state can destroy evolution: the elastic wave speed satisfies $\rho c^2=(\mu_p/\lambda)\khat^{\!\top}\A_0\khat$, so a negative directional stretch produces growth unbounded with wavenumber \citep{Rutkevich1969,JosephRenardySaut1985}. A Newtonian solvent changes the principal part and regularizes this loss of evolution \citep{JosephSaut1990,GuillopeSaut1990}; artificial stress diffusion is a distinct numerical regularization and is not interchangeable with solvent viscosity \citep{SureshkumarBeris1995}. For solvent-regularized Oldroyd-B, the exact frozen roots and their normalized principal-axis threshold are already known \citep{OwensPhillips2002,Chen2014}. Those results show why the inviscid sign test cannot simply be carried into a solution model: solvent damping permits a finite negative directional stretch. We use that result as the linear-spring baseline and make no priority claim for it.

The unresolved directional question begins when the constitutive tangent is not scalar. A least-eigenvalue test or a principal-axis projection can diagnose how far a numerical state has left the SPD cone, but it need not identify the physical wavevector direction that first destabilizes the coupled momentum--conformation symbol. In Giesekus, deformation-dependent mobility couples diagonal and shear perturbations away from a principal axis; equilibrium-normalized FENE-P supplies a useful contrast because its stress and relaxation tangents share a Peterlin factor. The distinction matters for high-Weissenberg-number solvers: a direction-blind alarm cannot distinguish a principal projection from a growing oblique sector. Resolving the full angle also tests a common extrapolation, namely that increasing the companion stretch indefinitely postpones instability along the least principal direction. The present analysis supplies that local spectral distinction, not a tolerance rule, a nonlinear saturation model or a guarantee for an inhomogeneous flow.

The claim must also be separated from adjacent precedents. General constitutive stability constraints for Giesekus in its classical admissible SPD domain are established by \citet{KwonLeonov1995}, and the mobility model itself originates with \citet{Giesekus1982}. \citet{RenardyRenardy1999} found oblique waves in two-layer Giesekus Couette flow, driven by an interfacial normal-stress jump; general shear-flow stability \citep{ConradOberlack2026} and local constitutive-Jacobian alignment \citep{Rezaee2026} concern different base states or different direction spaces. In particular, a direction in conformation-component space is not a physical wavevector angle, and an interfacial mode is not a bulk frozen-symbol mode. To our knowledge, these works do not classify the momentum-coupled, zero-frequency response of a single uniform supplied state outside the SPD cone over physical wavevector angle. Our novelty claim is restricted to that bulk, planar setting and not to oblique Giesekus instability in general.

For planar perturbations of a uniform frozen state, we reduce the complete coupled symbol to a static-compliance criterion and then use it to determine the first neutral direction. Under the registered model-specific tangent-stability conditions, $1+r\chi>0$ is equivalent to spectral stability and the first crossing is stationary (\S\ref{sec:compliance}). Oldroyd-B and equilibrium-normalized FENE-P select a principal direction, whereas Giesekus mobility can select an oblique one: at $\alpha=0.1$, $r=0.5$, $h=0$ and companion eigenvalue $\lambda_2=5$, full-angle onset is $\lambda_{1,c}=-1.933$ at $\theta_c=23.94^\circ$, before the principal value $-2.167$. For the same registered family, the full-angle stationary-neutral threshold approaches the finite limit $-2.319$ as the companion stretch grows (\S\ref{sec:oblique} and figure~\ref{fig:thresholds}). A rational deeper-state counterexample and independent matrix and uniform forced-base calculations provide checks that do not reuse the angular minimization. These results are linear, uniform-state and single-Fourier-mode statements; their three-dimensional and dynamical consequences are treated separately rather than folded into the central claim.
\section{A static-compliance criterion for the complete planar symbol}
\label{sec:compliance}

Directional minimization has spectral meaning only after the complete velocity--conformation symbol has been reduced without discarding a potentially unstable branch. A zero of a projected stationary coefficient would otherwise say nothing about an oscillatory root elsewhere in the spectrum. We therefore first identify the zero-frequency response that controls the coupled planar spectrum and establish when a stationary neutral surface is the first possible crossing. The reduction is performed at fixed non-zero wavenumber and retains inertia, stress diffusion and every constitutive spectator mode; only incompressibility and exact block elimination are used.

Consider an incompressible liquid with solvent viscosity $\mu_s>0$, polymer viscosity $\mu_p>0$, relaxation time $\lambda$, density $\rho$ and stress diffusivity $\kappa\geq0$:
\begin{align}
\rho(\partial_t\bm u+\bm u\!\cdot\!\nabla\bm u)
 &=-\nabla p+\mu_s\Delta\bm u+\nabla\!\cdot\!\bm\tau_p,
 &\nabla\!\cdot\!\bm u&=0,                                           \label{eq:mom}\\
\partial_t\A+\bm u\!\cdot\!\nabla\A
 &=(\nabla\bm u)\A+\A(\nabla\bm u)^{\!\top}
   -\lambda^{-1}\mathcal R(\A)+\kappa\Delta\A .                       \label{eq:conf}
\end{align}
Oldroyd-B uses $\lambda\bm\tau_p/\mu_p=\mathcal R(\A)=\A-\I$; equilibrium-normalized FENE-P uses $\lambda\bm\tau_p/\mu_p=\mathcal R(\A)=f\A-\I$, with $f=(b-2)/(b-\tr\A)$ and $f(2)=1$; Giesekus uses $\lambda\bm\tau_p/\mu_p=\A-\I$ and $\mathcal R(\A)=(\A-\I)+\alpha(\A-\I)^2$. These choices share the same momentum equation but respectively supply a linear tangent, a trace-coupled matched tangent and a quadratic mobility tangent. Let the uniform symmetric frozen state, not necessarily a steady or SPD state, have wave-frame components $a=\khat^{\!\top}\A_0\khat$, $c=\that^{\!\top}\A_0\khat$ and $d=\that^{\!\top}\A_0\that$, where $k\khat$ is a non-zero planar Fourier wavevector and $\that\perp\khat$. Incompressibility leaves one transverse velocity amplitude in this frame; its gradient forces the trace--shear constitutive block, while orthogonal constitutive components remain explicit spectator factors. We use
\begin{equation}
r=\frac{\mu_p}{\mu_s}>0,\qquad h=\kappa\lambda k^2\geq0,\qquad
\delta=\frac{\rho}{\mu_s\lambda k^2}\geq0,\qquad z=\lambda\sigma .     \label{eq:nondim}
\end{equation}
Here $h$ and $\delta$ refer to the specified $k$; no equivalence between diffusion and solvent regularization is implied. The frozen-state construction asks for the local linear response at the supplied $\A_0$ and does not require the unforced relaxation equation to hold there. Consequently, forcing needed to maintain a numerical base state is absent from the perturbation symbol, and the result is neither a global mode nor an inhomogeneous-flow calculation.

Eliminating the conformation amplitudes defines a transfer function $H(z)$ from transverse velocity to polymer shear response and the static compliance $\chi=H(0)$. The following statement includes all finite spectral branches, with stable spectator constitutive modes retained as factors.
\begin{proposition}[Static compliance and stationary onset]
\label{prop:compliance}
For every non-zero planar Fourier mode of the three models above, let $r>0$, $h\geq0$ and finite $\delta\geq0$, and suppose the uncoupled constitutive Fourier block is strictly stable under the model-specific conditions stated below. The complete finite symbol is Hurwitz stable if and only if
\begin{equation}
1+r\chi(\khat;\A_0,h)>0 .                                               \label{eq:mastercriterion}
\end{equation}
Equality gives a simple stationary zero root, whereas a negative margin gives a positive real root. Varying finite $\delta$ at fixed $k,h,r$ does not move the neutral surface; for $\delta=0$, stability refers to the reduced Stokes symbol.
\end{proposition}
Thus $1+r\chi$ is the zero-frequency directional stiffness of the coupled solvent--polymer response, not a realizability test on $\A_0$. Strict stability of the uncoupled constitutive block is essential: for FENE-P it includes the matched positive tangent conditions, and for Giesekus it excludes the constitutive-degeneracy boundary. The proposition is also pointwise in $k$ and direction. It neither compares different wavenumbers nor asserts that a path reaching constitutive degeneracy has first crossed the coupled neutral surface. These qualifications make the stationary statement a complete-symbol result without enlarging it into a general theorem for arbitrary constitutive tangents.

The non-zero-frequency exclusion is a feedback argument. For Giesekus set $b_G=1+h-2\alpha$, $p_i=b_G+2\alpha\lambda_i$, $g=(p_1+p_2)/2$, $P=b_G+2\alpha a$, $Q=b_G+2\alpha d$ and $K=aQ-2\alpha c^2$; the strict constitutive condition is $\alpha\geq0$ and $p_1>0$ for $\lambda_1\leq\lambda_2$, hence $g,p_2>0$. The feedback polynomial and transfer function are
\begin{align}
F_G(z)&=(1+\delta z)(z+p_1)(z+g)(z+p_2)+r(z+P)(az+K),                    \label{eq:giequartic}\\
H_G(z)&=\frac{(z+P)(az+K)}{(z+p_1)(z+g)(z+p_2)} .                       \label{eq:gie_transfer}
\end{align}
Dividing by the stable denominator gives $G(z)=1+\delta z+rH_G(z)$.
For $\lambda_1<0$, write $\xi=\cos^2\theta$, $\Delta=\lambda_2-\lambda_1$, $N=\Re H_G(i\omega)$ and $S=-\Im H_G(i\omega)/\omega$. Direct partial fractions give
\begin{equation}
N-gS=\frac{2\alpha\xi(1-\xi)\Delta^2
 [\omega^2+b_G^2-4\alpha^2\lambda_1\lambda_2]}
 {(p_1^2+\omega^2)(p_2^2+\omega^2)}\geq0 .                             \label{eq:nohopf_identity}
\end{equation}
A Hopf root would require simultaneously $N=-1/r<0$ and $S=\delta/r\geq0$, contradicting $N\geq gS\geq0$. Positivity of the square bracket follows directly: if $\lambda_2\geq0$ then $\lambda_1\lambda_2\leq0$, while if $\lambda_2<0$ the condition $p_1>0$ implies $b_G^2>4\alpha^2\lambda_1\lambda_2$. For $\lambda_1\geq0$, the numerator of $\Re H_G(i\omega)$ instead has a verified non-negative coefficient expansion in $\omega^2$ and $\sin^2\theta$, again excluding the Hopf requirement. Oldroyd-B closes by its quadratic Hurwitz condition; matched FENE-P closes by the cubic Routh determinant after its constant coefficient is written as the positive open-loop product times $1+r\chi_F$. A homotopy in coupling from $0$ to $r$ then cannot cross zero or a non-zero imaginary frequency while $1+q\chi>0$. The polynomial degree and leading coefficient remain fixed separately for $\delta>0$ and the reduced $\delta=0$ problem, so a root cannot change half-plane through infinity. If the margin is negative, the characteristic polynomial changes sign between $z=0$ and the positive real limit and therefore has a positive real root. At equality $G'(0)=\delta-rS(0)>0$, proving that the zero root is simple; the remaining roots are stable limits of the subcritical homotopy. Reduced polynomials or continuous extensions cover $\delta=0$, $\alpha=0$, $c=0$, equal eigenvalues and coincident constitutive poles; $p_1=0$ is excluded. Full residual identities, sector coefficients and degenerate branches are supplied in the appendices.

The three static responses expose the directional contrast. Oldroyd-B gives $\chi_{OB}=a/(1+h)$. For equilibrium-normalized FENE-P, let $f_0>0$, $g_0=f'(T_0)\geq0$, $\ell=f_0+g_0T_0>0$, $A_f=h+f_0$ and $B_f=h+\ell$; exact trace--shear elimination gives
\begin{equation}
\chi_F=\frac{f_0aB_f+2g_0c^2h}{A_fB_f},\qquad
\lambda_{1,c}^{F}=-\frac{h+f_0}{rf_0}.                                 \label{eq:chif}
\end{equation}
The shared non-negative tangent factor keeps the angular minimum principal and recovers $-1/r$ when $h=0$. More generally, this conclusion holds for a two-dimensional stress and relaxation tangent sharing the same scalar factor with $f_0>0$ and $g_0\geq0$; it is not asserted for unmatched closures or a negative derivative. Giesekus instead gives
\begin{equation}
\chi_G=\frac{PK}{p_1gp_2},\qquad K=aQ-2\alpha c^2,                      \label{eq:chig}
\end{equation}
so the wave-frame shear enters at order $\alpha$ and can move the minimum into the angular interval. The denominator also depends on the principal eigenvalues through the three stable tangent rates, making the competition between diagonal stretch and shear coupling genuinely directional. When $\alpha=0$ the shear correction disappears, the coincident spectator poles reduce continuously, and the Oldroyd-B principal selection is recovered. The criterion is therefore common, but the constitutive tangent selects the direction. Section~\ref{sec:oblique} now minimizes these exact responses rather than a principal projection.
\section{Constitutive selection of the first unstable direction}
\label{sec:oblique}

We now minimize the registered compliance over planar wavevector angle. Here ``first'' has a precise pathwise meaning: $|k|$, $\alpha$, $r$, $h$, $\delta$ and $\lambda_2$ are fixed, while $\lambda_1<\lambda_2$ decreases from a stable state, and the first full-angle stationary-neutral contact is recorded. Equivalently, the directional margin $M(\lambda_1)=\min_{0\leq\theta\leq\pi/2}[1+r\chi(\theta)]$ remains positive before contact and reaches zero at onset; an instability is asserted only on crossing to $M<0$. The half-quadrant contains all distinct directions because the frozen state is diagonal in its eigenframe and the compliances are even under reflection. With $\theta$ measured from the $\lambda_1$ eigenvector to $\khat$,
\begin{equation}
a=\lambda_1\cos^2\theta+\lambda_2\sin^2\theta,\qquad
c=(\lambda_2-\lambda_1)\sin\theta\cos\theta .                         \label{eq:angle}
\end{equation}
The Oldroyd-B and equilibrium-normalized FENE-P expressions are minimized at $\theta=0$, so they select the least principal direction. In those controls, rotating away from the least eigenvector replaces negative directional stretch by the larger companion component, and the matched FENE-P shear correction cannot overturn that ordering. The $c^2$ term in $\chi_G$ has the opposite capability: it permits an interior minimum even while both endpoints remain stable. Giesekus mobility can therefore select an oblique first instability under the same complete-symbol criterion. Because $\delta$ is absent from $\chi$, this direction and neutral depth are unchanged by finite inertia at fixed $h$, although the off-neutral eigenvalues and growth rates are not.

For the reference parameters
\begin{equation}
\alpha=0.1,\qquad r=0.5,\qquad h=0,qquad \lambda_2=5,                  \label{eq:referencecase}
\end{equation}
principal-axis analysis gives $\lambda_{1,c}^{\mathrm{principal}}=-13/6=-2.166666667$. Full-angle minimization of the interior quadratic, followed by the neutral condition, instead yields
\begin{equation}
\lambda_{1,c}=\frac{-1624+648\sqrt6}{19}=-1.933191930,qquad
\theta_c=23.938182^\circ .                                             \label{eq:exactonset}
\end{equation}
This is the first stationary-neutral contact as $\lambda_1$ decreases, and both principal directions are still stable there. The selected angle is measured from the $\lambda_1$ eigenvector, not from the shear or transverse direction. Relative to the magnitude of the principal value, $(|\lambda_{1,c}^{\mathrm{principal}}|-|\lambda_{1,c}|)/|\lambda_{1,c}^{\mathrm{principal}}|=10.8\%$: the principal calculation permits that much additional negative depth before signalling onset. This discrepancy is not a small correction to the selected angle; it changes which Fourier direction crosses first. It also persists for every finite $\delta$ because the same zero-frequency margin controls the neutral surface. The comparison concerns the registered fixed-wavenumber path, not a fastest-growing direction across wavenumbers or a global three-dimensional mode. Figure~\ref{fig:directional} places the control-model principal onset and the Giesekus oblique onset in the same $(\theta,\lambda_1)$ plane.

The departure from the endpoint is continuous. By a principal-to-oblique bifurcation we mean that the minimizer of $\chi_G(\theta)$ moves continuously from the endpoint $\theta=0$ into $0<\theta<\pi/2$, rather than switching between sampled angles. Writing the compliance as a quadratic in $s=\sin^2\theta$ makes the geometry explicit: below bifurcation its constrained minimizer is the endpoint $s=0$, at bifurcation the one-sided angular curvature vanishes, and above it the unconstrained minimizer enters $0<s<1$. The endpoint loses minimality at
\begin{equation}
\lambda_{2,\mathrm{bif}}
=\frac{b_G}{2\alpha}\left(\sqrt{1+\frac{r}{\alpha}}-2\right),
\qquad b_G=1+h-2\alpha,                                                 \label{eq:directionbif}
\end{equation}
for the positive-companion branch with $r>3\alpha$. At the reference $\alpha$, $r$ and $h$, $\lambda_{2,\mathrm{bif}}=1.797958971$; above it, the minimizing angle grows from zero and the oblique branch controls the full-angle neutral threshold. At the bifurcation itself the minimizing direction is still principal and its one-sided angular curvature is zero: $\lambda_{2,\mathrm{bif}}$ is a companion-stretch boundary, not a selected non-zero angle. The interior angle develops only after that boundary is crossed. The formula also shows which ingredients are required: $\alpha$ supplies the non-scalar mobility, $r$ transmits it through momentum coupling and $h$ shifts the tangent rates through $b_G$. The word bifurcation here describes the argmin geometry of the stationary-neutral problem, not a time-dependent branching solution or a claim about nonlinear states. Figure~\ref{fig:thresholds}(b) shows this endpoint--interior geometry continuously, rather than as a switch between sampled angles.

An exact deeper-state counterexample separates oblique growth from the marginal minimization used to locate onset. Retain $\alpha=0.1$, $r=0.5$ and $h=0$, take the Stokes limit $\delta=0$, and set $\lambda_1=-2$, $\lambda_2=5$ and $\cos^2\theta=23/28$. Direct substitution into the unreduced wave-frame block shows that both principal spectra are stable, whereas the oblique feedback polynomial is
\begin{equation}
z^3+\frac{117}{40}z^2+\frac{1277}{800}z-\frac{53}{1000}=0,              \label{eq:counterpoly}
\end{equation}
with positive real root $z=0.0313791668$. The negative constant coefficient and positive leading coefficient already force a positive real root; the quoted value is the direct root of the printed cubic. No angular optimizer, interpolation or figure is needed for this sign conclusion. The rational direction is therefore a hand-checkable falsification of the implication ``stable principal directions imply full-angle stability''. Figure~\ref{fig:directional}(b) locates this rational state inside the negative-margin sector; the printed cubic, not the colour field, supplies the positive root. Its angle is not the critical angle in \eqref{eq:exactonset}: the former is a deeper-state witness at $\lambda_1=-2$, whereas $23.938182^\circ$ is the first neutral direction on the registered path. Nor does the counterexample assert that every indefinite Giesekus state has an oblique unstable sector; it establishes existence within the registered family.

\begin{figure}
\centering
\includegraphics[width=\textwidth]{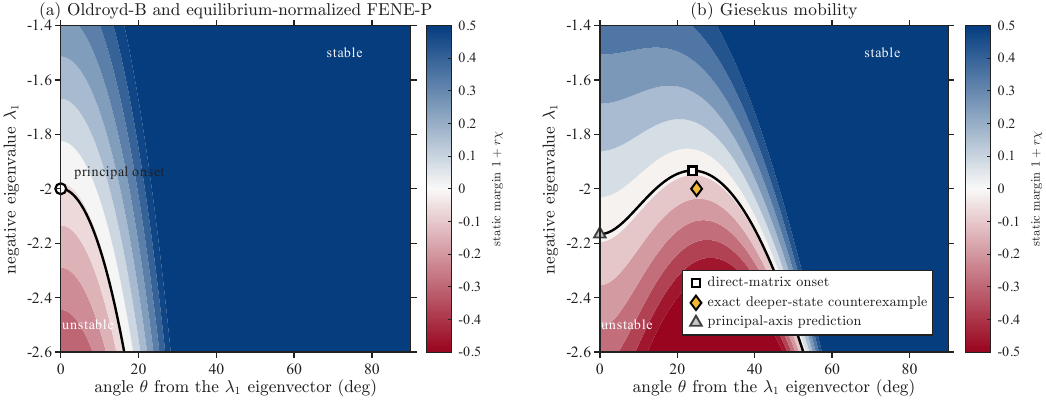}
\caption{Planar static-margin landscapes at $r=0.5$, $h=0$ and $\lambda_2=5$. Colours are dense evaluations of the analytic margin $1+r\chi$; the black contour is neutral. (a) Oldroyd-B and equilibrium-normalized FENE-P first reach the contour on the $\lambda_1$ principal direction. (b) Giesekus mobility raises an interior part of the contour to the direct-matrix onset (square) in \eqref{eq:exactonset}, before the principal-axis prediction (triangle). The diamond marks the exact deeper-state counterexample in \eqref{eq:counterpoly}. These are parameter-space landscapes, not flow-field snapshots or DNS.}
\label{fig:directional}
\end{figure}

The same angular geometry prevents unbounded principal-axis tolerance at large companion stretch. For $r>0$, $h\geq0$ and $0<\alpha<(1+h)/2$, minimize over angle at each finite $\lambda_2$ before taking $\lambda_2\to\infty$. This order matters: taking the principal direction first suppresses the $c^2$ mobility term whose increasing companion stretch drives the interior minimum. Balancing the leading diagonal and shear contributions gives a finite limiting negative eigenvalue and a non-zero limiting angle. The stationary-neutral branch approaches
\begin{equation}
\lambda_{1,c}^{(\infty)}
=\frac{b_G[\alpha-\sqrt{\alpha(\alpha+r)}]}{r\alpha},\qquad
\sin^2\theta_c^{(\infty)}
=-\frac{\alpha\lambda_{1,c}^{(\infty)}}{b_G}.                           \label{eq:saturation}
\end{equation}
At $\alpha=0.1$, $r=0.5$ and $h=0$, these limits are $\lambda_{1,c}^{(\infty)}=-2.319183588$ and $\theta_c^{(\infty)}=32.576263^\circ$. The full-angle threshold thus saturates at a finite oblique value, approached continuously by the finite-$\lambda_2$ neutral branch. By contrast, the principal-axis neutral formula continues to recede without bound and, for this family beyond $\lambda_2=16$, is only a formal extrapolation because it has already passed the constitutive-tangent coercivity boundary. This distinction prevents the mathematical continuation of one projected formula from being read as increasing physical or numerical tolerance. The limiting pair in \eqref{eq:saturation} is tied to the stated $\alpha,r,h$ family and to stationary planar minimization; changing model parameters changes both numbers, and no three-dimensional global selection theorem follows. Figure~\ref{fig:thresholds} shows the finite threshold and angle together, replacing the former table of sampled companion stretches.

Independent calculations check both the printed neutral curve and its stable/unstable sides. The first path assembles the full velocity--conformation matrices directly, without evaluating the closed-form compliance during the eigenvalue solve. Thirty-five such crossings agree with the analytic thresholds to $4.73\times10^{-13}$, and the largest registered neutral-root residual is $6.62\times10^{-14}$. Figure~\ref{fig:thresholds} displays the 15 crossings that match its three plotted $r$ families; the complete table also varies $\alpha$ and $h$. The second path evolves uniformly forced frozen states pseudospectrally at $\lambda_2=5,20,100,1000$. For each companion value, rates are measured at $\lambda_{1,c}+0.05$ on the stable side and $\lambda_{1,c}-0.05$ on the unstable side along the analytic critical angle. All twelve registered runs reproduce the predicted signs; two-sided interpolation differs from the closed-form curve by at most $1.08\times10^{-4}$. Grid, time step and seed were then varied one at a time at four registered anchors, and the largest rate change is below $7.8\times10^{-4}$ relative to the predicted rate, well below the preassigned comparison tolerance. The four open circles in figure~\ref{fig:thresholds}(a) report the interpolation, not extra fitted theory. These calculations verify the printed uniform symbol and its implementation; they are not inhomogeneous-flow DNS, physical validation or evidence of nonlinear activation. Figures~\ref{fig:directional} and \ref{fig:thresholds} separate the analytic selection mechanism from the non-source-identical threshold checks.

\begin{figure}
\centering
\includegraphics[width=\textwidth]{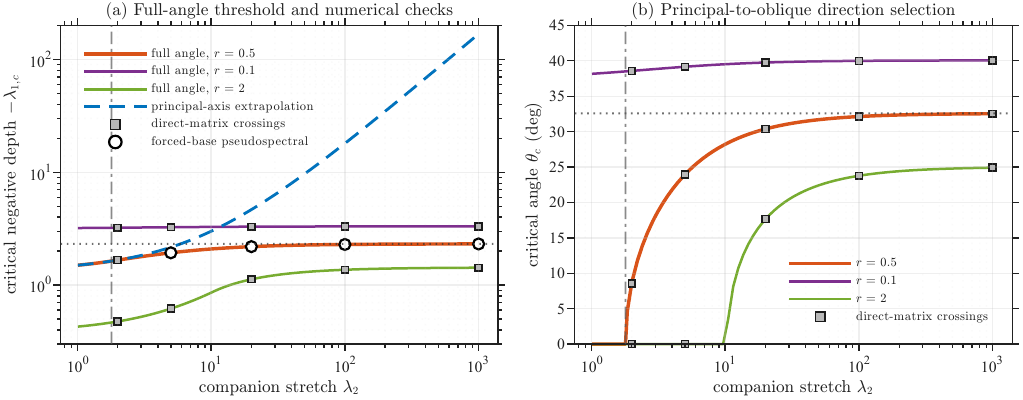}
\caption{Giesekus threshold and direction over the registered companion-stretch families at $\alpha=0.1$ and $h=0$. (a) Full-angle thresholds remain finite, while the formal principal-axis extrapolation recedes. (b) The selected angle leaves the endpoint continuously and approaches the finite limit in \eqref{eq:saturation}. Lines are analytic. Grey squares are 15 independent direct-matrix crossings, five for each plotted $r$ family. The four open circles in (a) are thresholds interpolated from 12 uniform forced-base pseudospectral runs for $r=0.5$; their largest deviation from the analytic curve is $1.08\times10^{-4}$. The vertical line marks \eqref{eq:directionbif}. No markers represent non-uniform-flow DNS.}
\label{fig:thresholds}
\end{figure}

\section{Consequences and dimensional scope}
\label{sec:consequences}
The coupled neutral surface, constitutive-tangent coercivity and nonlinear
relaxation basin mark three different failures.  The neutral surface asks when
solvent--polymer feedback first admits a growing Fourier mode; coercivity asks
whether the uncoupled linearized constitutive block still decays; and the basin
asks whether nonlinear relaxation returns the frozen eigenvalues to equilibrium.
For fixed $\lambda_2\geq0$, $r>0$ and $2\alpha<1+h$, decreasing $\lambda_1$
reaches the coupled Giesekus neutral surface while the constitutive tangent is
still coercive.  The detailed separation and its small-$r$ asymptotics are given
in the appendices.  This ordering does not locate the nonlinear
basin: its position relative to the other two boundaries is not universal.
Consequently, decay of the frozen coupled symbol neither establishes nonlinear
recoverability nor makes an indefinite conformation admissible.
The dimensional extension has two distinct evidential levels.  Exactly, when
the wavevector of a diagonal three-dimensional frozen state is constrained to a
specified eigenplane, the full system contains an invariant block identical to
the corresponding two-dimensional eigenvalue-pair problem.  Permuting the axes
gives the other eigenplanes, so the three-dimensional spectrum contains these
planar branches and cannot be more stable than the planar calculation predicts.
For the reference natural embedding $\lambda_3=1$, the complementary
out-of-plane principal branch crosses exactly at $\lambda_1=-1.5$, before the
planar oblique value $-1.933191930$; the latter must therefore not be reported as
the global onset of that three-dimensional embedding.  Separately, the finite
full-sphere searches performed here were consistent with an eigenplane pairwise
reduction, and two tested natural embeddings had a numerically leading oblique
$e_1$--$e_3$ branch.  These observations show that obliquity is not confined to
the original plane, but they do not prove that every global minimizer lies in an
eigenplane or establish a general three-dimensional pairwise criterion.
These conclusions remain local to a uniform frozen coefficient, linear
perturbations and one non-zero Fourier wavenumber of fixed magnitude.  They do
not demonstrate activation in an inhomogeneous flow, a spatial global mode or
wavepacket, or nonlinear saturation.  In
particular, allowing $\boldsymbol A_0$ to lie outside the positive-definite cone
is a diagnostic device, not a change in the microstructural admissibility of a
conformation tensor.  The analysis states how the registered symbol responds if
such an excursion is interrogated; it neither licenses the excursion as physical
nor prescribes how a numerical flow returns from it.

\section{Conclusions}
\label{sec:conclusions}
For non-zero planar Fourier modes of a uniform frozen state, static compliance
reduces the registered coupled velocity--conformation symbols to a directional
stability test.  Within the strictly stable constitutive sector, the sign of
$1+r\chi$ determines the complete planar spectrum: neutrality is stationary,
and inertia changes the roots without moving the neutral surface.  The resulting
direction is constitutive-model dependent.  Oldroyd-B and
equilibrium-normalized FENE-P retain principal selection, whereas Giesekus
mobility can select an oblique first instability.  At the reference parameters,
full-angle onset occurs at $\lambda_1=-1.933191930$ and
$\theta=23.938182^{\circ}$, before the principal value $-2.166666667$.  For the
registered large-stretch family, the oblique threshold approaches the finite
limit $-2.319183588$ at $32.576263^{\circ}$, while the formal principal-axis
neutral extrapolation has no finite limit.
The practical implication is deliberately narrower than a positivity rule.  A
least eigenvalue records the depth of a conformation excursion, but by itself it
does not identify the wavevector that grows; a directional compliance does.
Likewise, coupled onset, constitutive-tangent coercivity and nonlinear recovery
should be monitored as separate questions.  Three-dimensional eigenplane blocks
guarantee that the planar branches remain in the larger spectrum, while the
present full-sphere searches provide only tested, not theorem-level, information
about global selection.  None of these statements makes an indefinite
conformation microstructurally physical or establishes instability in an
inhomogeneous nonlinear flow.  Within that evidence envelope, the central
conclusion is that the first unstable direction, rather than the least
eigenvalue alone, is the quantity passed from an indefinite frozen state to the
coupled solver.

\medskip
\noindent\textbf{Funding.} This work was supported by the National Natural
Science Foundation of China (No.~12101527), the Natural Science Foundation of
Hunan Province (No.~2026JJ60005), the Science and Technology Innovation Program
of Hunan Province (No.~2026RC3172), and the 111 Project (No.~D23017).
Computational resources were provided by the High Performance Computing Platform
of Xiangtan University.

\smallskip
\noindent\textbf{Declaration of interests.} The author reports no conflict of
interest.

\smallskip
\noindent\textbf{Data availability statement.} The machine-readable data and
scripts supporting the uniform-symbol checks and figure generation are
available from the corresponding author upon reasonable request. Analytic
results follow from the equations in the article and appendices.

\smallskip
\noindent\textbf{Use of AI tools.} On 16 July 2026, the author used the
then-current hosted versions of Claude (Anthropic) and OpenAI Codex for drafting,
editing and symbolic checks.  The author verified all text, derivations and
numerical values and takes full responsibility for the article.

\clearpage
\appendix
These appendices are self-contained at the level of the uniform Fourier symbol.
Appendices~\ref{sup:setup}--\ref{sup:gieproof} prove the complete planar
stationary-onset criterion, including the finite-inertia and Stokes branches,
simple neutrality, the coupling homotopy and all registered degeneracies.
Appendix~\ref{sup:direction} gives the directional minimization and the finite
large-companion limit.  Appendices~\ref{sup:boundaries} and \ref{sup:gain}
separate recovery boundaries from time-dependent amplification.
Appendix~\ref{sup:threeD} distinguishes an exact eigenplane-restricted
three-dimensional substructure from finite full-sphere numerical searches.
The last section records verification tolerances and, equally importantly,
what those checks do not establish.

\section{Planar Fourier reduction and scope of the criterion}
\label{sup:setup}

Let $\A_0$ be a spatially uniform, symmetric frozen state.  For a non-zero
two-dimensional Fourier mode, write
\[
(\delta\bm u,\delta\A,\delta p)
 =(\bm v,\bm E,\pi)\exp(i\bm k\boldsymbol\cdot\bm x+\sigma t),
\qquad k=|\bm k|>0.
\]
Set $\bm k=k\khat$ and choose $\that\perp\khat$.  Incompressibility leaves
$\bm v=v\that$.  In the wave frame,
\begin{equation}
\A_0=\begin{pmatrix}a&c\\c&d\end{pmatrix},\qquad
\bm E=\begin{pmatrix}x&w\\w&y\end{pmatrix}.                               \label{eq:Srot}
\end{equation}
The stretching source and the only transverse stress component seen by
momentum are
\begin{equation}
(\nabla\delta\bm u)\A_0+\A_0(\nabla\delta\bm u)^\T
=ikv\begin{pmatrix}0&a\\a&2c\end{pmatrix},
\qquad \that\boldsymbol\cdot\delta\bm\tau\,\khat .                         \label{eq:Sstretch}
\end{equation}
Pressure absorbs the longitudinal force.

Use the dimensionless quantities
\begin{equation}
z=\lambda\sigma,\qquad
h=\kappa\lambda k^2\geq0,\qquad
\delta=\frac{\rho}{\mu_s\lambda k^2}\geq0,\qquad
r=\frac{\mu_p}{\mu_s}>0,\qquad U=i\lambda kv .                            \label{eq:Sscales}
\end{equation}
For each registered model, elimination of the stable, uncoupled
constitutive spectators gives
\begin{equation}
G(z)=1+\delta z+rH(z)=0,\qquad \chi=H(0),                                 \label{eq:Sfeedback}
\end{equation}
where $H$ is the constitutive shear response to unit $U$.  At
$\delta=0$, the momentum equation is an algebraic Stokes constraint; the
finite spectrum is defined only after that constraint has been eliminated.

\begin{proposition}[complete planar symbol at fixed wavenumber]
\label{prop:Scriterion}
Fix $k>0$, $r>0$, $h\geq0$ and finite $\delta\geq0$.  Assume that the
uncoupled constitutive Fourier block is strictly stable.  For the
equilibrium-normalized FENE-P convention, also assume the matched
stress--relaxation tangent described in \S\ref{sup:fene}, with
$f_0>0$, $g_0\geq0$ and $\ell=f_0+g_0T_0>0$.  For Giesekus, assume
$\alpha\geq0$ and
\[
p_1=1+h-2\alpha+2\alpha\lambda_1>0,
\qquad \lambda_1\leq\lambda_2.
\]
Then, for every fixed wave direction,
\begin{enumerate}
\item the complete finite velocity--conformation spectrum is strictly stable
      if and only if $1+r\chi>0$;
\item $1+r\chi=0$ gives one simple root $z=0$, while all other finite roots
      remain in $\Rea z<0$;
\item $1+r\chi<0$ gives at least one positive real root.
\end{enumerate}
Consequently no non-zero-frequency crossing precedes the registered neutral
surface, and changing finite inertia $\delta$ at fixed $k,h,r$ does not move
that surface.
\end{proposition}

The word first requires a separate path convention.  At fixed
$k,\alpha,r,h,\delta,\lambda_2$, consider a continuous decrease of
$\lambda_1$ that remains inside the strict constitutive-stability open set,
and define
\[
\mathcal M(\lambda_1)=\min_{\theta}\{1+r\chi(\lambda_1,\theta)\}.
\]
The first value at which $\mathcal M=0$ is the first neutral contact; it is
called a first instability only if the path continues into
$\mathcal M<0$.  An angular tangency, simultaneous directions, or a path
that reaches $p_1=0$ first must be reported separately.

\section{Oldroyd-B and equilibrium-normalized FENE-P}
\label{sup:obfene}

\subsection{Oldroyd-B}

Only $v$ and $w$ feed back:
\begin{equation}
(\rho\sigma+\mu_sk^2)v=\frac{\mu_p}{\lambda}ikw,\qquad
(\sigma+\lambda^{-1}+\kappa k^2)w=ikav .                                 \label{eq:Sobpair}
\end{equation}
With $s_0=1+h>0$, elimination gives
\begin{equation}
F_{\rm OB}(z)=(1+\delta z)(z+s_0)+ra=0,\qquad
H_{\rm OB}(z)=\frac{a}{z+s_0},\qquad
\chi_{\rm OB}=\frac{a}{s_0}.                                             \label{eq:Sobpoly}
\end{equation}
For $\delta>0$,
\[
F_{\rm OB}(z)=\delta z^2+(1+\delta s_0)z+s_0+ra .
\]
The first two coefficients are positive, so the quadratic is Hurwitz exactly
when $s_0+ra=s_0(1+r\chi_{\rm OB})>0$.  For $\delta=0$ the reduced Stokes
polynomial is the corresponding first-degree expression.  A negative
constant term gives a positive real root by continuity on the positive real
axis.  At neutrality,
\begin{equation}
F'_{\rm OB}(0)=1+\delta s_0>0,                                           \label{eq:Sobsimple}
\end{equation}
so the zero root is simple.  The remaining planar conformation component is
a stable triangular spectator with pole $-s_0$.

\subsection{Matched FENE-P trace--shear closure}
\label{sup:fene}

Let $T=\tr\A$ and consider a scalar spring factor $f(T)$.  The tangent of
$f\A$ is
\begin{equation}
D(f\A)_{\A_0}[\bm E]
=f_0\bm E+g_0(\tr\bm E)\A_0,\qquad
f_0=f(T_0),\quad g_0=f'(T_0).                                           \label{eq:Sftangent}
\end{equation}
The result below uses the same factor in the polymer stress and relaxation
tangent.  It is therefore a matched stress--relaxation statement, not a
claim for an arbitrary scalar spring law.  For the equilibrium-normalized
two-dimensional FENE-P convention
\[
f(T)=\frac{b-2}{b-T},\qquad T_0<b,
\]
one has $f_0>0$, $g_0=(b-2)/(b-T_0)^2>0$ and
$\ell=f_0+g_0T_0>0$.

Let $s=x+y=\tr\bm E$.  The momentum, shear and trace equations close exactly:
\begin{align}
(1+\delta z)U&=-r(f_0w+g_0cs),                                         \label{eq:SfeneU}\\
(z+h+f_0)w+g_0cs&=aU,                                                   \label{eq:Sfenew}\\
(z+h+\ell)s&=2cU.                                                       \label{eq:Sfenes}
\end{align}
No omitted planar conformation component feeds back into these variables.
Writing
\[
A_f=h+f_0>0,\qquad B_f=h+\ell>0,
\]
gives
\begin{equation}
F_F(z)=(1+\delta z)(z+A_f)(z+B_f)
+r\{f_0a(z+B_f)+2g_0c^2(z+h)\}=0.                                     \label{eq:Sfenepoly}
\end{equation}
At zero frequency,
\begin{equation}
\chi_F=\frac{f_0a}{A_f}+\frac{2g_0c^2h}{A_fB_f},\qquad
F_F(0)=A_fB_f(1+r\chi_F).                                               \label{eq:Sfenechi}
\end{equation}

For $\delta>0$, write $F_F=a_3z^3+a_2z^2+a_1z+a_0$, where
\begin{align}
a_3&=\delta,\qquad
a_2=1+\delta(A_f+B_f),                                                   \label{eq:Sfenecoefs1}\\
a_1&=A_f+B_f+\delta A_fB_f+r(f_0a+2g_0c^2),                              \label{eq:Sfenecoefs2}\\
a_0&=A_fB_f+r(f_0aB_f+2g_0c^2h).                                       \label{eq:Sfenecoefs3}
\end{align}
If $a_0>0$, then
\[
rf_0a>-A_f-\frac{2rg_0c^2h}{B_f},
\]
and hence
\begin{align}
a_1&>B_f+\delta A_fB_f+\frac{2rg_0c^2\ell}{B_f}>0,                       \label{eq:Sfenea1}\\
(A_f+B_f)a_1-a_0
&>A_fB_f+B_f^2+\delta A_fB_f(A_f+B_f)
+2rg_0c^2\ell\left(1+\frac{A_f}{B_f}\right)>0.                          \label{eq:Sfenedelta}
\end{align}
It follows that
\[
a_2a_1-a_3a_0
=a_1+\delta\{(A_f+B_f)a_1-a_0\}>0,
\]
so all cubic Routh--Hurwitz conditions hold.  When $\delta=0$,
\eqref{eq:Sfenepoly} is reduced before analysis to a quadratic; the same
bounds make both non-leading coefficients positive.  If $a_0<0$, a positive
real root exists.  At $a_0=0$, substitution of the equality gives
\begin{equation}
a_1=B_f+\delta A_fB_f+\frac{2rg_0c^2\ell}{B_f}>0,                        \label{eq:Sfenesimple}
\end{equation}
so the zero root is simple and the remaining quadratic factor is stable.
The trace--shear orthogonal spectator has pole $-A_f$.  The proof fails
without the shared tangent or when the sign conditions on $f_0,g_0,B_f$ are
removed.

For later directional comparison, let
$\lambda_1\leq\lambda_2$, $\Delta=\lambda_2-\lambda_1$, and
$x_\theta=\sin^2\theta$, with $\theta$ measured from the
$\lambda_1$ eigenvector.  Then
\begin{equation}
a=\lambda_1+\Delta x_\theta,\qquad
c^2=\Delta^2x_\theta(1-x_\theta),                                       \label{eq:Sangle}
\end{equation}
and
\begin{equation}
\chi_F(x_\theta)=\frac{f_0}{A_f}(\lambda_1+\Delta x_\theta)
+\frac{2g_0h\Delta^2}{A_fB_f}x_\theta(1-x_\theta).                       \label{eq:Sfeneangle}
\end{equation}
This is concave on $[0,1]$.  Its minimum is at an endpoint, and the
$x_\theta=0$ endpoint is no larger because $f_0\Delta/A_f\geq0$.
Thus the matched FENE-P problem retains principal-axis selection.

\section{Giesekus matrix, transfer function and Stokes reduction}
\label{sup:giematrix}

The Giesekus relaxation tangent, including the Fourier diffusion shift, is
\begin{equation}
\mathcal L[\bm E]
=(1+h)\bm E+\alpha\{(\A_0-\I)\bm E+\bm E(\A_0-\I)\}.                    \label{eq:Sgietangent}
\end{equation}
Define
\begin{equation}
b_G=1+h-2\alpha,\quad
P=b_G+2\alpha a,\quad Q=b_G+2\alpha d,\quad
g=b_G+\alpha(a+d),\quad
p_i=b_G+2\alpha\lambda_i.                                               \label{eq:Sgiedefs}
\end{equation}
The planar equations are
\begin{align}
zx&=-Px-2\alpha cw,                                                       \label{eq:Sgiex}\\
zw&=aU-\alpha cx-gw-\alpha cy,                                           \label{eq:Sgiew}\\
zy&=2cU-2\alpha cw-Qy,                                                    \label{eq:Sgiey}\\
(1+\delta z)U&=-rw.                                                       \label{eq:Sgiemom}
\end{align}
For $\delta>0$, the matrix for $(U,x,w,y)^\T$ is
\begin{equation}
\bm M_\delta=
\begin{pmatrix}
-\delta^{-1}&0&-r\delta^{-1}&0\\
0&-P&-2\alpha c&0\\
a&-\alpha c&-g&-\alpha c\\
2c&0&-2\alpha c&-Q
\end{pmatrix}.                                                           \label{eq:Sgiematrix}
\end{equation}
Rotational invariance gives
\begin{equation}
P+Q=p_1+p_2=2g,\qquad
PQ-4\alpha^2c^2=p_1p_2.                                                  \label{eq:Sgieinvariants}
\end{equation}
Eliminating $x$ and $y$ yields
\begin{equation}
\frac{w}{U}
=H(z)
=\frac{(z+P)(az+K)}
{(z+p_1)(z+g)(z+p_2)},\qquad
K=aQ-2\alpha c^2,                                                        \label{eq:Sgietransfer}
\end{equation}
and therefore
\begin{equation}
F_G(z)=(1+\delta z)(z+p_1)(z+g)(z+p_2)
+r(z+P)(az+K)=0.                                                        \label{eq:Sgiepoly}
\end{equation}
In particular,
\begin{equation}
\chi_G=\frac{PK}{p_1gp_2},\qquad
F_G(0)=p_1gp_2(1+r\chi_G).                                               \label{eq:Sgiechi}
\end{equation}
For $\alpha\geq0$ and $\lambda_1\leq\lambda_2$, $p_1>0$ implies
$g,p_2>0$ and is equivalent to strict stability of the three uncoupled
constitutive poles.

The Stokes case is not obtained by retaining the singular first row of
\eqref{eq:Sgiematrix}.  Set $\delta=0$, impose $U=-rw$, and reduce to
$X=(x,w,y)^\T$:
\begin{equation}
\frac{\dd X}{\dd(t/\lambda)}
=\bm M_0X,\qquad
\bm M_0=
\begin{pmatrix}
-P&-2\alpha c&0\\
-\alpha c&-(g+ra)&-\alpha c\\
0&-2c(\alpha+r)&-Q
\end{pmatrix}.                                                          \label{eq:Sgiestokesmatrix}
\end{equation}
Direct expansion gives
\[
\det(z\I-\bm M_0)
=(z+p_1)(z+g)(z+p_2)+r(z+P)(az+K),
\]
the cubic obtained from \eqref{eq:Sgiepoly} after the algebraic velocity
constraint has been removed.

\section{Giesekus exclusion of a prior Hopf crossing}
\label{sup:gieproof}

Stable open-loop poles alone do not exclude a feedback Hopf crossing.  For
example, $H(z)=1/(z+1)^3$, $r=8$ and $\delta=0$ give
$(z+1)^3+8=0$ with roots $\pm i\sqrt3$.  The following model-specific
frequency identity is therefore essential.

\subsection{The sector \texorpdfstring{$\lambda_1<0$}{lambda1 < 0}}

Let $\xi=\cos^2\theta$ and $\Delta=\lambda_2-\lambda_1$.  For distinct
poles, the transfer function has the partial fractions
\begin{equation}
H(z)=\frac{R_1}{z+p_1}+\frac{R_g}{z+g}+\frac{R_2}{z+p_2},                 \label{eq:Spartial}
\end{equation}
where
\begin{equation}
R_1=2\lambda_1\xi(1-\xi),\quad
R_g=(2\xi-1)\{\lambda_1\xi-\lambda_2(1-\xi)\},\quad
R_2=2\lambda_2\xi(1-\xi).                                               \label{eq:Sresidues}
\end{equation}
For $\omega\neq0$, define
\begin{equation}
\mathscr S(\omega)=-\frac{\Ima H(i\omega)}{\omega}
=\sum_{j\in\{1,g,2\}}\frac{R_j}{p_j^2+\omega^2},\qquad
\mathscr N(\omega)=\Rea H(i\omega)
=\sum_{j\in\{1,g,2\}}\frac{R_jp_j}{p_j^2+\omega^2},                     \label{eq:SNS}
\end{equation}
where $p_g=g$.  Cross-multiplication gives the identity
\begin{equation}
\boxed{\;
\mathscr N-g\mathscr S
=\frac{2\alpha\xi(1-\xi)\Delta^2
\{\omega^2+b_G^2-4\alpha^2\lambda_1\lambda_2\}}
{(p_1^2+\omega^2)(p_2^2+\omega^2)}
\geq0 .\;}                                                              \label{eq:Sidentity}
\end{equation}
The bracket is strictly positive.  If $\lambda_2\geq0$, then
$\lambda_1\lambda_2\leq0$.  If $\lambda_2<0$, $p_1>0$ gives
$b_G>-2\alpha\lambda_1$, and hence
$b_G^2>4\alpha^2\lambda_1^2\geq
4\alpha^2\lambda_1\lambda_2$.

At a hypothetical root $z=i\omega$, $\omega\neq0$, the real and imaginary
parts of \eqref{eq:Sfeedback} require
\begin{equation}
\mathscr N=-\frac1r<0,\qquad
\mathscr S=\frac{\delta}{r}\geq0.                                       \label{eq:Shopfreq}
\end{equation}
These conditions contradict
$\mathscr N\geq g\mathscr S\geq0$.

\subsection{The sector \texorpdfstring{$\lambda_1\geq0$}{lambda1 >= 0}}

Put
\[
L=\lambda_1\geq0,\qquad D=\lambda_2-\lambda_1\geq0,\qquad
s=\sin^2\theta,\qquad p=p_1>0,\qquad X=\omega^2.
\]
Exact collection gives
\begin{equation}
\Rea H(i\omega)
=\frac{C_2X^2+C_1X+C_0}
{(p_1^2+X)(g^2+X)(p_2^2+X)},                                            \label{eq:SSPDreal}
\end{equation}
with
\begin{align}
C_2={}&L(p+\alpha D)+Dps+D^2\alpha s(3-2s),                              \label{eq:SC2}\\
C_1=2\{&
2D^4\alpha^3s^2
+2D^3L\alpha^3(3s^2-3s+1)
+D^3\alpha^2p(3s^2+s)\nonumber\\
&+2D^2L\alpha^2p(3s^2-3s+2)
+3D^2\alpha p^2s+3DL\alpha p^2+D p^3s+Lp^3\},                            \label{eq:SC1}\\
C_0={}&p(p+\alpha D)(p+2\alpha D)(p+2\alpha Ds)
\{2DL\alpha(1-s)+Dps+Lp\}.                                               \label{eq:SC0}
\end{align}
Every term is non-negative on $s\in[0,1]$; in particular
$3s^2-3s+1>0$ and $3s^2-3s+2>0$.  Thus
$\Rea H(i\omega)\geq0$, contradicting the Hopf requirement
$\Rea H=-1/r$.  This sector also cannot be neutral at zero frequency because
$H(0)\geq0$.

\subsection{Coupling homotopy, simple neutrality and degeneracies}

Replace $r$ by a homotopy parameter $q\in[0,r]$.  At $q=0$, all
constitutive poles and, for $\delta>0$, the inertial pole
$-1/\delta$ are strictly stable.  If $1+r\chi>0$, then
$1+q\chi>0$ throughout the homotopy: this is immediate for $\chi\geq0$,
and for $\chi<0$ it follows from
$1+q\chi\geq1+r\chi$.  Hence no zero root occurs.  The preceding frequency
argument, with $r$ replaced by $q$, excludes every non-zero imaginary root.
The degree and leading coefficient remain fixed separately on the
$\delta>0$ quartic and the $\delta=0$ reduced cubic, so no root changes
half-plane through infinity.  This proves sufficiency.  Necessity follows
because $F_G(0)<0$ when $1+r\chi<0$, whereas
$F_G(x)\to+\infty$ as $x\to+\infty$, giving a positive real root.

At equality, neutrality lies in the $\lambda_1<0$ sector.  Taking
$\omega\to0$ in \eqref{eq:Sidentity} and using
$\mathscr N(0)=H(0)=-1/r$ gives
\[
\mathscr S(0)\leq-\frac{1}{rg}<0.
\]
Since $H'(0)=-\mathscr S(0)$,
\begin{equation}
G'(0)=\delta+rH'(0)=\delta-r\mathscr S(0)>0,                              \label{eq:Sgiesimple}
\end{equation}
so the zero root is simple.  The other roots are limits of the stable
$q<r$ homotopy and remain strictly in the left half-plane.

The partial-fraction display is only a convenient representation.
Equation~\eqref{eq:Sidentity} is a cross-multiplied polynomial identity, so
coincident constitutive poles and the endpoints $\xi=0,1$ follow by
continuous extension or direct substitution.  If $\alpha=0$, the feedback
factor is the Oldroyd-B quadratic with stable repeated spectators.  If
$c=0$, $\lambda_1=\lambda_2$, or $\delta=0$, the corresponding reduced
equations above give the same result without division by a vanishing
quantity.  The surface $p_1=0$ is excluded: it is the constitutive
degeneracy boundary, not part of the strict open set in
Proposition~\ref{prop:Scriterion}.

\section{Directional selection and the finite large-companion limit}
\label{sup:direction}

Let $x=\sin^2\theta$ and $\Delta=\lambda_2-\lambda_1$.  For Giesekus,
\begin{equation}
P=p_1+2\alpha\Delta x,\qquad
Q=p_2-2\alpha\Delta x,\qquad
c^2=\Delta^2x(1-x).                                                      \label{eq:Sgieangle1}
\end{equation}
The numerator contains the cancellation
\begin{equation}
K=aQ-2\alpha c^2
=\lambda_1p_2+b_G\Delta x.                                               \label{eq:SgieKsimple}
\end{equation}
Because $p_1gp_2$ is independent of direction, minimizing $\chi_G$ is
equivalent to minimizing the quadratic
\begin{equation}
PK=(p_1+2\alpha\Delta x)
(\lambda_1p_2+b_G\Delta x).                                              \label{eq:Sgiequad}
\end{equation}
Inside the indefinite constitutive-stable sector, $p_1>0$ forces $b_G>0$,
so the quadratic is convex.  Its interior stationary point is
\begin{equation}
x_*=-\frac{2\alpha\lambda_1p_2+b_Gp_1}
{4\alpha b_G\Delta},                                                     \label{eq:Sxstar}
\end{equation}
when this value lies in $(0,1)$; otherwise an endpoint controls.

Setting $x_*=0$ simultaneously with the endpoint neutral condition gives
the endpoint-to-interior transition
\begin{equation}
\lambda_{2,\mathrm{bif}}
=\frac{b_G}{2\alpha}
\left(\sqrt{1+\frac r\alpha}-2\right).                                   \label{eq:Sdirectionbif}
\end{equation}
Here bifurcation means a continuous change in the minimizing argument from
the endpoint into the open interval; it is not a dynamical Hopf
bifurcation.  When $r>3\alpha$ the transition occurs at positive companion
stretch.  If $r\leq3\alpha$, its formal value is non-positive and the
positive-companion branch is already interior-controlled.

For $\alpha=0.1$, $r=0.5$, $h=0$ and $\lambda_2=5$, substituting the
stationarity relation into $F_G(0)=0$ gives
\begin{equation}
19\lambda_1^2+3248\lambda_1+6208=0.                                     \label{eq:Squadratic}
\end{equation}
The first root in the constitutive-stable sector is
\begin{equation}
\lambda_{1,c}=\frac{-1624+648\sqrt6}{19}
=-1.933191930\ldots,\qquad
\theta_c=23.938182^\circ.                                                \label{eq:Sexactcritical}
\end{equation}
The principal-axis value is $-13/6=-2.166666667\ldots$, which would allow
about $10.8\%$ more negative depth under the registered comparison
$\{|\lambda_{1,c}^{\rm principal}|-|\lambda_{1,c}|\}/
|\lambda_{1,c}^{\rm principal}|$.

The deeper state $\lambda_1=-2$, $\lambda_2=5$ admits an exact rational
check at $\cos^2\theta=23/28$:
\begin{align}
a&=-\frac34,\quad d=\frac{15}{4},\quad
c=\frac{\sqrt{115}}4,\quad
p_1=\frac25,\quad g=\frac{11}{10},\quad p_2=\frac95,                    \label{eq:Scounter1}\\
P&=\frac{13}{20},\quad Q=\frac{31}{20},\quad
K=-\frac{13}{5},\quad
F_G(0)=-\frac{53}{1000}.                                                 \label{eq:Scounter2}
\end{align}
Thus a positive real root exists although both principal spectra remain
stable.  The Stokes positive root is
$z=0.0313791668\ldots$; this deep-state angle
$24.997400^\circ$ is not the onset angle in
\eqref{eq:Sexactcritical}.

\subsection{Large companion stretch}

Fix $r>0$, $h\geq0$ and
$0<\alpha<(1+h)/2$, so $b_G>0$.  At each finite
$\lambda_2$, first minimize \eqref{eq:Sgiequad} over the full angular
interval and select the first stationary-neutral root.  Substitution of the
interior stationary point into $F_G(0)=0$, division by
$\lambda_2^2$, and then $\lambda_2\to+\infty$ gives
\begin{equation}
\alpha r\lambda_1^2-2\alpha b_G\lambda_1-b_G^2=0.                        \label{eq:Slimitquadratic}
\end{equation}
The two algebraic roots are
\[
\lambda_\pm
=\frac{b_G\{\alpha\pm\sqrt{\alpha(\alpha+r)}\}}{r\alpha}.
\]
The plus root has positive $\lambda_1$ and gives
$x_\infty=-\alpha\lambda_+/b_G<0$, outside the angular interval.
The admissible negative branch is therefore
\begin{equation}
\lambda_{1,c}^{(\infty)}
=\frac{b_G\{\alpha-\sqrt{\alpha(\alpha+r)}\}}{r\alpha},\qquad
\sin^2\theta_c^{(\infty)}
=-\frac{\alpha\lambda_{1,c}^{(\infty)}}{b_G}
=\frac{1}{1+\sqrt{1+r/\alpha}}.                                        \label{eq:Slimit}
\end{equation}
This root selection uses both negative critical depth and
$0<x_\infty<1$; constitutive stability alone would not discard the plus
root.

At the reference values,
\begin{equation}
\lambda_{1,c}^{(\infty)}
=\frac85(1-\sqrt6)
=-2.319183588\ldots,\qquad
\theta_c^{(\infty)}=32.576263^\circ.                                    \label{eq:Slimitref}
\end{equation}
The principal-axis neutral formula is
\begin{equation}
\lambda_{1,c}^{\rm principal}
=-\frac{b_G+\alpha\lambda_2}{r+\alpha}
=-\frac{\lambda_2+8}{6}\longrightarrow-\infty.                          \label{eq:Sprincipalformal}
\end{equation}
For the reference family it crosses the coercivity boundary
$\lambda_1=-4$ once $\lambda_2>16$; beyond that point it is only a formal
principal-axis neutral extrapolation, not a stable-sector threshold.

\begin{table}[t]
\centering
\small
\begin{tabular}{lccc}
\toprule
changed parameter & fixed parameters & $\lambda_{1,c}$ & $\theta_c$ \\
\midrule
$\alpha=0.05$ & $r=0.5,\ h=0,\ \lambda_2=20$ & $-3.204690824$ & $18.20736^\circ$ \\
$\alpha=0.20$ & $r=0.5,\ h=0,\ \lambda_2=20$ & $-1.035156935$ & $35.76695^\circ$ \\
$r=0.10$      & $\alpha=0.1,\ h=0,\ \lambda_2=20$ & $-3.293505963$ & $39.75743^\circ$ \\
$r=2.00$      & $\alpha=0.1,\ h=0,\ \lambda_2=20$ & $-1.125592017$ & $17.63813^\circ$ \\
$h=0.50$      & $\alpha=0.1,\ r=0.5,\ \lambda_2=20$ & $-3.454808722$ & $29.02037^\circ$ \\
$h=2.00$      & $\alpha=0.1,\ r=0.5,\ \lambda_2=20$ & $-6.886291608$ & $25.00838^\circ$ \\
\bottomrule
\end{tabular}
\caption{Representative full-angle planar stationary-neutral calculations.
They demonstrate parameter dependence and the persistence of interior
selection in these cases, not universal oblique selection.}
\label{tab:Srobust}
\end{table}

\begin{table}[t]
\centering
\small
\begin{tabular}{rcccc}
\toprule
$\lambda_2$ & $\lambda_{1,c}$ &
$|\lambda_{1,c}-\lambda_{1,c}^{(\infty)}|$ &
$\theta_c$ & $|\theta_c-\theta_c^{(\infty)}|$ \\
\midrule
$10^3$ & $-2.31636498151$ & $2.819\times10^{-3}$ &
$32.53167994^\circ$ & $4.458\times10^{-2}\,{}^\circ$ \\
$10^6$ & $-2.31918076332$ & $2.825\times10^{-6}$ &
$32.57621839^\circ$ & $4.460\times10^{-5}\,{}^\circ$ \\
$10^7$ & $-2.31918330594$ & $2.825\times10^{-7}$ &
$32.57625853^\circ$ & $4.460\times10^{-6}\,{}^\circ$ \\
\bottomrule
\end{tabular}
\caption{High-precision checks of the analytic limit
\eqref{eq:Slimitref}.  The sequence is a verification of the formula, not
the proof that the limit is finite.}
\label{tab:Slimitcheck}
\end{table}

\section{Recovery, tangent coercivity and the coupled boundary}
\label{sup:boundaries}

The three boundaries below answer different questions.  The relaxation basin
concerns the nonlinear base trajectory, tangent coercivity concerns the
uncoupled constitutive linearization, and $1+r\min_\theta\chi_G=0$ concerns
the momentum-coupled Fourier symbol.

Under Oldroyd-B relaxation alone,
\[
\dot\A=-(\A-\I)/\lambda.
\]
The right-hand side commutes with $\A$, so the eigenframe is fixed and
\begin{equation}
a_i(t)=1+\{a_i(0)-1\}e^{-t/\lambda}.                                     \label{eq:Sobrelax}
\end{equation}
When diffusion is retained, the complete two-dimensional determinant balance
is
\begin{equation}
\frac{D\det\A}{Dt}
=2(\nabla\boldsymbol\cdot\bm u)\det\A
-\frac{2}{\lambda}\det\A+\frac{\tr\A}{\lambda}
+\kappa\,\operatorname{adj}(\A):\Delta\A .                              \label{eq:Sdetfull}
\end{equation}
The last term vanishes for a uniform state or $\kappa=0$, but has no fixed
sign in general.  The coefficient $2/\lambda$ in this scalar balance is not
an eigenvalue crossing rate.

For Giesekus relaxation alone, each eigenvalue satisfies
\begin{equation}
\frac{\dd a}{\dd(t/\lambda)}
=-(a-1)\{1+\alpha(a-1)\}.                                                \label{eq:Sgiescalar}
\end{equation}
For $\alpha>0$,
\begin{equation}
a(t)-1
=\frac{(a_0-1)e^{-t/\lambda}}
{1+\alpha(a_0-1)(1-e^{-t/\lambda})}.                                    \label{eq:Sgiesolution}
\end{equation}
The basin of attraction of $a=1$ is
\begin{equation}
a_0>1-\frac1\alpha.                                                       \label{eq:Sbasin}
\end{equation}
Equality is the second equilibrium; below it the denominator in
\eqref{eq:Sgiesolution} vanishes at finite positive time and the solution
blows down rather than returning to one.

The Fourier-shifted tangent rates in an eigenbasis are
\begin{equation}
\gamma_{ij}=1+h+\alpha(\lambda_i+\lambda_j-2).                            \label{eq:Stangentrates}
\end{equation}
The smallest is $p_1=1+h-2\alpha+2\alpha\lambda_1$, so tangent coercivity
requires
\begin{equation}
\lambda_1>\lambda_1^*
=1-\frac{1+h}{2\alpha}.                                                  \label{eq:Stangentboundary}
\end{equation}

\begin{lemma}[the coupled surface precedes coercivity in the registered sector]
\label{lem:Sordering}
Fix $\lambda_2\geq0$, $r>0$, and
$0<\alpha<(1+h)/2$.  As $\lambda_1$ decreases within the coercive sector,
the first coupled stationary-neutral contact occurs at a value
$\lambda_{1,c}>\lambda_1^*$.  For small $r$,
\begin{equation}
\lambda_{1,c}-\lambda_1^*
=\frac{r(1+h-2\alpha)}{8\alpha^2}+O(r^2),\qquad
\sin^2\theta_c\longrightarrow\frac12.                                  \label{eq:Smarginlead}
\end{equation}
\end{lemma}

\begin{proof}
At fixed $\lambda_2\geq0$, let $\lambda_1\downarrow\lambda_1^*$, so
$p_1\to0^+$.  With
$\Delta_*=\lambda_2-\lambda_1^*>0$, the exact angular identities give
\[
P=2\alpha\Delta_*x,\quad
K=2\alpha\Delta_*(1-x)\lambda_1^*,\quad
p_2=2\alpha\Delta_*,\quad g=\alpha\Delta_* .
\]
Consequently,
\begin{equation}
\lim_{p_1\to0^+}p_1\chi_G
=2x(1-x)\lambda_1^*.                                                     \label{eq:Schiresidue}
\end{equation}
Because $\lambda_1^*<0$, the angular minimum occurs asymptotically at
$x=1/2$ and
$\min_\theta\chi_G\sim\lambda_1^*/(2p_1)\to-\infty$.  At $\lambda_1=0$,
$K=b_G\lambda_2x\geq0$, so the coupled margin is positive.  Continuity
therefore places its first zero strictly above $\lambda_1^*$.  Solving
$1+r\lambda_1^*/(2p_{1,c})=0$ and using
$p_1=2\alpha(\lambda_1-\lambda_1^*)$ gives
\eqref{eq:Smarginlead}.
\end{proof}

The restriction $\lambda_2\geq0$ is part of the statement; the older,
broader condition $\lambda_2\geq\lambda_1$ is not used.  The remaining two
boundaries obey the exact difference
\begin{equation}
\lambda_1^*-\left(1-\frac1\alpha\right)
=\frac{1-h}{2\alpha}.                                                    \label{eq:Sboundarydiff}
\end{equation}
Hence, for $h<1$, decreasing $\lambda_1$ meets the coupled surface, then
coercivity, then the relaxation-basin boundary.  At $h=1$ the last two
coincide.  For $h>1$ the basin lies above coercivity, and its order relative
to the coupled surface is not universal.  Comparing
\eqref{eq:Smarginlead} with \eqref{eq:Sboundarydiff} gives only the
small-$r$ crossover estimate
$r\sim4\alpha(h-1)/(1+h-2\alpha)$; it is not a global ordering formula.

\section{Amplification on freely relaxing backgrounds}
\label{sup:gain}

\subsection{Exact Oldroyd-B propagator}

Let $\tau=t/\lambda$, $q=1+h$ and consider a spatially uniform
Oldroyd-B base with no velocity gradient.  Along a fixed principal
direction,
\[
a(\tau)=1+(a_0-1)e^{-\tau}.
\]
In the Stokes limit, $U=-rw$, and the shear perturbation obeys
\begin{equation}
\frac{\dd w}{\dd\tau}=-\{q+ra(\tau)\}w.                                  \label{eq:Swode}
\end{equation}
Thus any non-zero amplitude proportional to $w$ satisfies
\begin{equation}
\log\frac{|w(\tau)|}{|w(0)|}
=r(1-a_0)(1-e^{-\tau})-(r+q)\tau.                                      \label{eq:Sgainhistory}
\end{equation}
Define
\[
Z=\frac{r(1-a_0)}{r+q}.
\]
If $Z\leq1$, the maximum is the initial amplitude.  If $Z>1$, the peak is
at $\tau_*=\log Z$ and
\begin{equation}
\log G_{\max}=(r+q)(Z-1-\log Z).                                        \label{eq:Sloggain}
\end{equation}
For $a_0<0$, recovery occurs at
$\tau_{\rm heal}=\log(1-a_0)$, so
\begin{equation}
\tau_{\rm heal}-\tau_*=\log\frac{r+q}{r},                                \label{eq:Sgainlag}
\end{equation}
independent of the initial depth.  For a prescribed $G_c>1$, put
$C=\log G_c/(r+q)$.  The branch $Z>1$ is
\begin{equation}
Z_c=-W_{-1}(-e^{-1-C}),\qquad
a_{0,c}=1-\frac{r+q}{r}Z_c.                                             \label{eq:SLambert}
\end{equation}
In the diffusion-free small-solvent convention
$r=(1-\beta)/\beta$, $q=1$, and for $a_0=-d$ with $d,\beta\ll1$,
\begin{equation}
\log G_{\max}
=\frac{(d-\beta)^2}{2\beta}
+o\!\left(\frac{(d-\beta)^2}{\beta}\right).                             \label{eq:Snear}
\end{equation}
These formulas concern a uniform single Fourier mode, not a spatial global
mode or an inhomogeneous-flow energy norm.

\subsection{Registered Giesekus coordinate-norm calculation}
\label{sup:giesgain}

For Giesekus, retain the Stokes restriction $h=\delta=0$ used in the
registered calculation.  The base eigenframe remains fixed and
\begin{equation}
\lambda_i(\tau)-1
=\frac{\{\lambda_i(0)-1\}e^{-\tau}}
{1+\alpha\{\lambda_i(0)-1\}(1-e^{-\tau})}.                              \label{eq:Sgielamtime}
\end{equation}
Choose one wave direction at $\tau=0$ and keep that direction fixed relative
to the eigenframe.  With $X=(x,w,y)^\T$, the exact time-dependent Stokes
system is
\begin{equation}
\frac{\dd X}{\dd\tau}=\bm M(\tau)X,\qquad
\bm M(\tau)=
\begin{pmatrix}
-P&-2\alpha c&0\\
-\alpha c&-(g+ra)&-\alpha c\\
0&-2c(\alpha+r)&-Q
\end{pmatrix}_{\A_0=\A_0(\tau)}.                                       \label{eq:SgieLTV}
\end{equation}
Let $\bm\Phi'=\bm M(\tau)\bm\Phi$, $\bm\Phi(0)=\I$.  The reported quantity is
\begin{equation}
G(\tau)=\|\bm\Phi(\tau,0)\|_2,\qquad
G_{\rm pre}=\max_{0\leq\tau\leq\tau_{\rm heal}}G(\tau),                 \label{eq:Sgiegain}
\end{equation}
where the norm is the coordinate Euclidean norm of $(x,w,y)$, not a
coordinate-invariant energy norm.  For $\lambda_1(0)<0$,
\begin{equation}
\tau_{\rm heal}
=\log\frac{(1-\alpha)\{1-\lambda_1(0)\}}
{1-\alpha+\alpha\lambda_1(0)}.                                         \label{eq:Sgiehealtime}
\end{equation}

The archived computation used DOP853 with relative tolerance $10^{-10}$,
absolute tolerance $10^{-12}$, chunked state-transition renormalization and
integration to $\tau_{\rm heal}+4$ to confirm post-peak decay.  At each
registered state the oblique direction is the initial static-compliance
minimizer; it is not re-optimized for transient gain.

\begin{table}[t]
\centering
\small
\begin{tabular}{ccccccc}
\toprule
$\lambda_1(0)$ & direction & $\theta_0$ & frozen $z$ &
$G_{\rm pre}$ & $\tau_{\rm peak}$ & $\tau_{\rm heal}$ \\
\midrule
$-1.95$ & principal & $0^\circ$ & $-0.13000000$ & $1.0000000$ & $0$ & $1.326002$ \\
$-1.95$ & initial $\chi$ minimum & $24.2098^\circ$ & $0.00801290$ &
$1.1471359$ & $0.421826$ & $1.326002$ \\
$-2.00$ & principal & $0^\circ$ & $-0.10000000$ & $1.0000000$ & $0$ & $1.349927$ \\
$-2.00$ & initial $\chi$ minimum & $24.9974^\circ$ & $0.03137917$ &
$1.1671697$ & $0.443360$ & $1.349927$ \\
\bottomrule
\end{tabular}
\caption{Archived coordinate $2$-norm gains for
$\alpha=0.1$, $r=0.5$, $h=\delta=0$ and $\lambda_2(0)=5$.
The principal value $1$ is asserted only for these two registered states.}
\label{tab:Sgiesgain}
\end{table}

These values are supporting, partial evidence.  The archived Oldroyd-B
degeneration agrees with \eqref{eq:Sloggain} to relative errors between
$9.6\times10^{-9}$ and $3.6\times10^{-7}$, but an independent RK45
recalculation of the two target values $1.1471359$ and $1.1671697$ has not
been permanently archived.  Accordingly Table~\ref{tab:Sgiesgain} is not
used as a main-text claim and must not be described as energy amplification,
a safety bound, a universal principal-axis monotonicity result, or a DNS
consequence.  The source table and algorithm are available from the
corresponding author upon reasonable request.

\section{Three-dimensional eigenplane blocks and numerical full-sphere tests}
\label{sup:threeD}

\subsection{Exact result for a prescribed eigenplane}
\label{sup:threeDexact}

Let
\[
\A_0=\operatorname{diag}(\lambda_1,\lambda_2,\lambda_3)
\]
and constrain $\khat$ to the $(\e_1,\e_2)$ eigenplane.  With
$\that$ the in-plane transverse direction, the wave-frame base has the
block
\[
\begin{pmatrix}a&c&0\\c&d&0\\0&0&\lambda_3\end{pmatrix}.
\]
The three-dimensional velocity has an in-plane polarization $\that$ and an
out-of-plane polarization $\e_3$.  The full symbol separates exactly into:
\begin{enumerate}
\item the $(U,x,w,y)$ in-plane block of
      \eqref{eq:Sgiex}--\eqref{eq:Sgiemom}, identical entry by entry to the
      two-dimensional $(\lambda_1,\lambda_2)$ problem and independent of
      $\lambda_3$;
\item an out-of-plane block for
      $U_o$, $m=\delta A_{3k}$ and $n=\delta A_{3t}$;
\item the scalar $\delta A_{33}$ spectator with pole
      $-p_3=-(b_G+2\alpha\lambda_3)$.
\end{enumerate}
The out-of-plane block is
\begin{align}
(1+\delta z)U_o&=-rm,                                                     \label{eq:Soutmom}\\
(z+P_3)m+\alpha cn&=aU_o,                                                 \label{eq:Soutm}\\
\alpha cm+(z+Q_3)n&=cU_o,                                                 \label{eq:Soutn}
\end{align}
where
\[
P_3=b_G+\alpha(a+\lambda_3),\qquad
Q_3=b_G+\alpha(d+\lambda_3).
\]
Its characteristic polynomial is
\begin{equation}
F_{\rm out}(z)
=(1+\delta z)\{(z+P_3)(z+Q_3)-\alpha^2c^2\}
+r\{a(z+Q_3)-\alpha c^2\}.                                              \label{eq:Soutpoly}
\end{equation}
Define $g_{ij}=b_G+\alpha(\lambda_i+\lambda_j)$.  At zero frequency,
\begin{equation}
\chi_{\rm out}(\theta)
=\frac{aQ_3-\alpha c^2}{P_3Q_3-\alpha^2c^2}
=\frac{\lambda_1}{g_{13}}\cos^2\theta
+\frac{\lambda_2}{g_{23}}\sin^2\theta.                                  \label{eq:Soutchi}
\end{equation}
In the strict constitutive sector this is linear and non-decreasing in
$\sin^2\theta$, so its minimum is the principal direction $\khat=\e_1$.
The exact additional eigenplane branch is neutral at
\begin{equation}
\lambda_1
=-\frac{b_G+\alpha\lambda_3}{r+\alpha}.                                  \label{eq:Soutthreshold}
\end{equation}
Permuting the axes gives an exact copy of the corresponding two-dimensional
pair block in every prescribed eigenplane.  Therefore the
three-dimensional spectrum contains all such planar branches and cannot be
more stable than those branches predict.  This inclusion statement does not
say that an arbitrary global minimizer over the wavevector sphere must lie
in an eigenplane.

At $\alpha=0.1$, $r=0.5$, $h=0$, $\lambda_2=5$ and the natural embedding
$\lambda_3=1$, \eqref{eq:Soutthreshold} gives the exact branch
\begin{equation}
\lambda_1=-\frac32=-1.5,                                                 \label{eq:Soutref}
\end{equation}
which precedes the planar oblique value $-1.933191930$.  Thus the latter is
not the global onset of this three-dimensional embedding.

\subsection{Finite full-sphere searches: numerical evidence only}
\label{sup:threeDnumeric}

The full-sphere calculations used an octant $91\times91$ angular mesh with
four local refinements and compared the numerical sphere minimum with the
lowest threshold predicted by the tested eigenvalue-pair branches.  The
seven parameter sets below are
\[
\begin{array}{c|ccccc}
\text{code}&\alpha&r&h&\lambda_2\\ \hline
R&0.1&0.5&0&5\\
A05&0.05&0.5&0&5\\
A20&0.2&0.5&0&5\\
R2&0.1&2&0&5\\
L10&0.1&0.5&0&10\\
H05&0.1&0.5&0.5&5\\
A30&0.3&1&0&8
\end{array}
\]
For each set, the labels $1$, mid and top mean
$\lambda_3=1$, $\lambda_3=(\lambda_1+\lambda_2)/2$ evaluated
self-consistently at threshold, and $\lambda_3=\lambda_2$, respectively.
The branch classification in Table~\ref{tab:S3Dsphere} comes from a probe
$0.02$ beyond the pairwise onset.  It is not a reported critical angle.

\begingroup
\small
\setlength{\LTleft}{0pt}
\setlength{\LTright}{0pt}
\begin{longtable}{@{}llrrrl@{}}
\caption{Finite full-sphere threshold searches.  The close agreement is
numerical evidence consistent with pairwise selection; it is not an exact
global three-dimensional theorem.}
\label{tab:S3Dsphere}\\
\toprule
set & $\lambda_3$ choice & pair prediction & sphere estimate &
$|\Delta|$ & probe branch \\
\midrule
\endfirsthead
\toprule
set & $\lambda_3$ choice & pair prediction & sphere estimate &
$|\Delta|$ & probe branch \\
\midrule
\endhead
R   & $1$   & $-1.500000000$ & $-1.50000000$ & $0$ & principal $e_1/e_3$ \\
R   & mid   & $-1.615384615$ & $-1.61538462$ & $1.6{\times}10^{-11}$ & principal $e_1/e_3$ \\
R   & top   & $-1.933191930$ & $-1.93319193$ & $5.8{\times}10^{-11}$ & degenerate ring \\
A05 & $1$   & $-1.727272727$ & $-1.72727273$ & $2.7{\times}10^{-11}$ & principal $e_1/e_3$ \\
A05 & mid   & $-1.782608696$ & $-1.78260870$ & $4.8{\times}10^{-11}$ & principal $e_1/e_3$ \\
A05 & top   & $-2.090909091$ & $-2.09090909$ & $9.1{\times}10^{-12}$ & principal, degenerate \\
A20 & $1$   & $-0.943756960$ & $-0.94375698$ & $2.2{\times}10^{-8}$ & oblique $e_1$--$e_3$ \\
A20 & mid   & $-0.977257101$ & $-0.97725719$ & $8.4{\times}10^{-8}$ & oblique $e_1$--$e_3$ \\
A20 & top   & $-1.010746251$ & $-1.01074625$ & $2.1{\times}10^{-12}$ & degenerate ring \\
R2  & $1$   & $-0.428571429$ & $-0.42857143$ & $2.9{\times}10^{-11}$ & principal $e_1/e_3$ \\
R2  & mid   & $-0.488372093$ & $-0.48837209$ & $2.3{\times}10^{-11}$ & principal $e_1/e_3$ \\
R2  & top   & $-0.619047619$ & $-0.61904762$ & $4.8{\times}10^{-11}$ & principal, degenerate \\
L10 & $1$   & $-1.500000000$ & $-1.50000000$ & $0$ & principal $e_1/e_3$ \\
L10 & mid   & $-1.876438017$ & $-1.87643822$ & $2.0{\times}10^{-7}$ & oblique $e_1$--$e_3$ \\
L10 & top   & $-2.089854771$ & $-2.08985477$ & $5.7{\times}10^{-12}$ & degenerate ring \\
H05 & $1$   & $-2.333333333$ & $-2.33333333$ & $3.3{\times}10^{-11}$ & principal $e_1/e_3$ \\
H05 & mid   & $-2.384615385$ & $-2.38461538$ & $1.5{\times}10^{-11}$ & principal $e_1/e_3$ \\
H05 & top   & $-2.917900918$ & $-2.91790092$ & $6.9{\times}10^{-11}$ & degenerate ring \\
A30 & $1$   & $-0.394446668$ & $-0.39444675$ & $8.1{\times}10^{-8}$ & oblique $e_1$--$e_3$ \\
A30 & mid   & $-0.419699669$ & $-0.41969973$ & $5.7{\times}10^{-8}$ & oblique $e_1$--$e_3$ \\
A30 & top   & $-0.426173065$ & $-0.42617307$ & $6.1{\times}10^{-11}$ & degenerate ring \\
\bottomrule
\end{longtable}
\endgroup

Across these 21 threshold rows the largest difference is
$2.1\times10^{-7}$, and a denser reference sweep had a maximum
sphere-versus-pairwise gap of $6.3\times10^{-7}$.  A separate direct
growth-rate check agreed in sign with the static-compliance prediction in
399 of 399 random draws.  In particular, the $A20$ and $A30$ natural
$\lambda_3=1$ embeddings had numerically leading oblique
$e_1$--$e_3$ branches.  These finite searches are consistent with
eigenplane pairwise selection and show that obliquity is not confined to
the original plane.  They do not prove a general pairwise reduction, an
if-and-only-if condition involving
$\min(\lambda_2,\lambda_3)$, a universal mobility range, or a global
critical angle.

\section{Verification records and claim boundary}
\label{sup:checks}

The symbolic proof and the finite numerical checks have different logical
roles.  Exact identities establish the statements in
Proposition~\ref{prop:Scriterion}; finite root scans are regression tests of
their implementation.  Direct matrices test the printed symbol through an
independent assembly.  Uniform forced-base pseudospectral calculations test
the time integrator and growth-rate extraction, but are not DNS of an
inhomogeneous flow.

\begin{table}[t]
\centering
\small
\begin{tabular}{rcccc}
\toprule
$\lambda_2$ & analytic $\lambda_{1,c}$ & measured $\lambda_{1,c}$
& stable-side $z$ & unstable-side $z$ \\
\midrule
5    & $-1.93319193$ & $-1.93308397$ & $-0.02390682$ & $+0.02401028$ \\
20   & $-2.19260442$ & $-2.19255008$ & $-0.01921742$ & $+0.01925924$ \\
100  & $-2.29157250$ & $-2.29155928$ & $-0.01765858$ & $+0.01766792$ \\
1000 & $-2.31636498$ & $-2.31635562$ & $-0.01728632$ & $+0.01729279$ \\
\bottomrule
\end{tabular}
\caption{Uniform forced-base pseudospectral brackets.  The stable state is
$\lambda_{1,c}+0.05$ and the unstable state is
$\lambda_{1,c}-0.05$.  The measured neutral value is a two-rate
interpolation, not a global-flow critical point.}
\label{tab:Spdesat}
\end{table}

The 35 directly assembled matrix crossings differ from the analytic
thresholds by at most $4.73\times10^{-13}$.  Across their neutral states and
$\delta=0,0.01,1$, the largest absolute leading root is
$6.62\times10^{-14}$.  In the four forced-base brackets of
Table~\ref{tab:Spdesat}, all eight stable/unstable signs are correct; the
largest interpolated-threshold error is $1.08\times10^{-4}$ and the largest
rate error relative to the direct matrix is $2.62\times10^{-4}$.

\begingroup
\small
\setlength{\LTleft}{0pt}
\setlength{\LTright}{0pt}
\begin{longtable}{@{}p{0.18\textwidth}p{0.25\textwidth}p{0.25\textwidth}p{0.25\textwidth}@{}}
\caption{Verification hierarchy and claim boundary.  Paths are relative to
the submission root.}
\label{tab:Sverification}\\
\toprule
record & registered result & supports & does not support \\
\midrule
\endfirsthead
\toprule
record & registered result & supports & does not support \\
\midrule
\endhead
L02 permanent symbolic replay &
15 symbolic or degeneracy identities have exact residual zero; finite scans
contain 1200 Oldroyd-B, 1600 FENE-P and 2400 Giesekus root cases, plus 240
Giesekus simple-neutral samples, with zero failures. &
Algebraic implementation of the model-specific proof and its reduced
branches. &
The finite scans alone are not a proof and do not cover arbitrary
constitutive laws or three-dimensional wavevectors.\\[2pt]

FENE-P convention checks &
Exact trace--shear closure and determinant residuals are zero; 400 of 400
selected eigenvalue tests agree with the sign of $1+r\chi_F$. &
The matched equilibrium-normalized convention and the printed planar
closure. &
An arbitrary scalar spring with $g_0<0$ or unmatched stress and relaxation
factors.\\[2pt]

Direct Giesekus matrices &
35 neutral thresholds agree within $4.73\times10^{-13}$; neutral roots across
three inertia values are below $6.62\times10^{-14}$ in magnitude. &
Independent assembly of the frozen finite-dimensional symbol and inertia
independence of the neutral surface in those cases. &
A nonlinear-flow, global-mode or residence-time result.\\[2pt]

Uniform forced-base brackets &
Eight of eight signs pass; maximum interpolated-threshold error
$1.08\times10^{-4}$. &
The uniform pseudospectral implementation and the analytic growth-rate sign. &
Four-roll-mill DNS, spontaneous base-state formation or nonlinear
saturation.\\[2pt]

T4 one-factor controls &
Across four anchors, the largest grid, time-step or seed drift relative to
$|z_{\rm th}|$ is $7.71\times10^{-4}$, below the registered $2\%$ gate. &
Local numerical sensitivity at those anchors. &
Global grid independence or a combined multi-factor error bound.\\[2pt]

T6 diffusion mapping &
Two $(k,\kappa)$ pairs with the same $h=0.5$ differ in measured rate by at
most $1.9\times10^{-4}$ relatively; the largest theory error is
$2.8\times10^{-4}$. &
Dependence on the Fourier diffusion shift $h=\kappa\lambda k^2$ in the tested
uniform cases. &
Cross-wavenumber global first-mode selection.\\[2pt]

Three-dimensional checks &
The exact prescribed-eigenplane block is verified symbolically; 21 finite
sphere thresholds are within $2.1\times10^{-7}$ of tested pair predictions. &
Exact eigenplane branch inclusion, plus finite numerical evidence consistent
with pairwise selection. &
An exact global three-dimensional pairwise theorem or a universal critical
angle.\\[2pt]

Giesekus gain archive &
The DOP853 archive gives the four rows in Table~\ref{tab:Sgiesgain}; the
Oldroyd-B degeneration is reproduced. &
The stated coordinate-norm calculation on the specified freely relaxing
background. &
Independent confirmation of the two target anchors, an energy norm, a safety
certificate or a DNS consequence.\\
\bottomrule
\end{longtable}
\endgroup

The machine-readable outputs and scripts supporting the registered checks in
these appendices are available from the corresponding author upon reasonable
request. No calculation here shows
that a non-positive conformation state is physically realizable.  The
claimed object remains the linear, uniform, single-mode symbol and its
registered supporting checks.

\end{document}